\documentclass[a4paper,12pt]{article}

\usepackage{graphicx, amssymb, latexsym, amsfonts, amsmath, lscape, amscd,
amsthm, color, epsfig}

\usepackage[T1]{fontenc}
\usepackage[latin1]{inputenc}

\usepackage[babel=true]{csquotes}

\usepackage{authblk}

\usepackage{setspace}

\usepackage{enumitem}

\newtheorem{theorem}{Theorem}

\newtheorem{lemma}[theorem]{Lemma}

\author{Mouhamad El Joubbeh}
\affil{Lebanese University, KALMA Laboratory, Beirut, Lebanon}
\date{}

\begin{document}

\title{Minimal Separators in Graphs}

\maketitle

\begin{abstract}
The Known Menger's theorem states that in a finite graph, the size of a minimum separator set of any pair of vertices is equal to the maximum number of disjoint paths that can be found between these two vertices. In this paper, we study the minimal separators of two non-adjacent vertices in a finite graph, and we give a new elementary proof of Menger's theorem.
\end{abstract}

\section{Introduction}

 Menger's theorem states that, the size of a minimum separator set of any pair of non-adjacent vertices is equal to the maximum number of disjoint paths that can be found between these two vertices. Menger's theorem was first proved by Karl Menger \cite{Menger} in 1927. 
Later on, many different shorter proofs  were given, as Menger's proof was considered a bit long and complicated. Before giving an idea about the proofs that were made, we state in the following some basic definitions and notations that were widely used in the attempts of proving Menger's statement, and that we will also adopt in our work in the latter section. For $u$ and $v$ being two vertices in a graph $G$, a set $S\subseteq V(G)-\{u,v\}$ is a $uv$-separator of $G$ if $u$ and $v$ lie in different components of $G-S$: that is, if every $uv$-path in $G$ contains a vertex in $S$. The minimum order of a $uv$-separator of $G$ is called the $uv$-connectivity of $G$ and is denoted by $\kappa_{G}(u,v)$. Note that if $uv\in E(G)$, then $G$ has no $uv$-separator, in this case we will consider $\kappa_{G}(u,v)=\infty$. A $uv$-separator $S$ of $G$ is said to be a minimal $uv$-separator of $G$ if $|S|=\kappa_{G}(u,v)$. As previously mentioned, a set of $uv$-paths is called internally disjoint if these paths are pairwise disjoint except for the vertices $u$ and $v$, and the maximum number of internally disjoint $uv$-paths in $G$ is denoted by $\mu_{G}(u,v)$. Since every $uv$-separator of $G$ must contain an internal vertex from each path in any set of internally disjoint $uv$-paths in $G$, then we obviously have $\mu_{G}(u,v)\leq \kappa_{G}(u,v) $.\\
After Menger proved his theorem, it was formulated and generalized by many ways, as by the Max-flow Min-cut theorem \cite{Fulkerson} in 1956, which is an elementary theorem within the field of network flows, that actually had some surprising implications in graph
theory. On the other hand, for shorter proofs of Menger's theorem that were established, the first one was given by G. A. Dirac \cite{Dirac} in 1966, where he proved the result by contradiction, after assuming that the statement of Menger is not true and working on a graph with minimal number of vertices not satisfying this statement. In 1978, Peter V. O'Neil \cite{Peter} took a different perspective while proving Menger's theorem, as the ones usually considered in proving its statement, as instead of finding a set of paths internally disjoint of cardinal equal to the cardinal of a considered minimal separator in a graph, he proved that there exists a separator of cardinal equal to the number of the maximum internally disjoint paths. Also, considering simpler proofs of Menger's result, there is one that was given by W. McCuaig\cite{McCuaig} in 1984 by using induction on the number of vertices of the separating set. It could also be interesting to refer that some researchers gave an equivalent formulation of Menger's Theorem \cite{Diestel}: For any two sets $V$ and $W$ of vertices in a graph $G$, a $VW$-path is a path from some vertex $v$ in $V$ to some vertex $w$ in $W$ that passes through no other vertices of $V$ and $W$. A set $S$ of vertices separates $V$ and $W$ if every $VW$-path contains a vertex of $S$, and $S$ is called a $VW$-separating set. It was proved that for any positive integer $k$, there are $k$ pairwise disjoint $VW$-paths in $G$ if and only if every $VW$-separating set contains at least $k$ vertices. Finally, the most recent proof of Menger's theorem was given by 
F. G\"oring \cite{Goring} in 2000.\\
In this paper, we study the minimal separators for it's own sake, we prove in particular that if $S$ is a minimal $uv$-separator in a graph $G$, then $\kappa_{G-e}(u,v)=\kappa_{G}(u,v)$ for all $e=xy$ where $x,y\in S$. This yields us to make a new proof of Menger's theorem.

\section{Minimal separator}
\begin{lemma}\label{lem 1}
	\rm{
		Consider a graph $G$, and let $u,v\in V(G)$ such that $uv\notin E(G)$. Then, $\kappa_{G}(u,v)-1\leq \kappa_{G-a}(u,v)\leq \kappa_{G}(u,v)$ for all $a\in (V(G)-\{u,v\})\cup E(G)$.
	}
\end{lemma} 

\begin{proof}
	Let $a\in (V(G)-\{u,v\})\cup E(G)$ and let $G'=G-a$. By simply remarking that any $uv$-separator in $G$ is a $uv$-separator in $G'$, then $\kappa_{G'}(u,v)\leq \kappa_{G}(u,v)$. In the other hand, if $\kappa_{G'}(u,v)< \kappa_{G}(u,v)-1$, then for any $uv$-separator $S$ in $G'$, $S\cup \{a\}$ wen $a$ is a vertex or $S\cup \{x\}$ wen $a=xy$ is an edge, is a $uv$-separator in $G$ with $|S|\leq \kappa_{G}(u,v)-1$, a contradiction.
\end{proof}

\begin{theorem} \label{thm 1}
	\rm{
		Consider a graph $G$, and let $u,v\in V(G)$ such that $uv\notin E(G)$. Then, $ \kappa_{G-e}(u,v)= \kappa_{G}(u,v) \ \ \forall \ e\in <S>$, where $S$ is a minimal $uv$-separator of $G$.
	}
\end{theorem}
\begin{proof}
	We will proceed by induction on $\kappa_{G}(u,v)$. The case $\kappa_{G}(u,v)=1$ being trivial. Now, for $\kappa_{G}(u,v)=k$, ($k\geq2$). Let $S=\{x_{1},x_{2}, ... , x_{k}\}$ be a minimal $uv$-separator of $G$. Suppose to the contrary that there exists an edge $e\in <S>$ such that $\kappa_{G-e}(u,v)\neq k$. We are looking for $uv$-path $P$ in $G$ such that $P\cap S=\phi$. This will give us the contradiction.\\
	\linebreak
	By using Lemma \ref{lem 1}, we have $\kappa_{G-e}(u,v)=k-1$. Without loss of generality we may suppose that $e=x_{1}x_{2}$. Let $G'=G-e$ and $S'$ be a minimal $uv$-separator of $G'$; $|S'|=\kappa_{G'}(u,v)=k-1$. The first observation of this analysis is that $S\cap S'=\phi$. Otherwise, let $x_{i}\in S\cap S'$, for $1\leq i \leq k$. Let $G_{i}=G-x_{i}$, $G'_{i}=G_{i}-e$ and $S_{i}=S-x_{i}$. It is clear that $\kappa_{G_{i}}(u,v)=k-1$, and $S_{i}$ is a minimal $uv$-separator of $G_{i}$. If $i\in\{1,2\}$, $G'_{i}=G_{i}$, so $\kappa_{G'_{i}}(u,v)=k-1$. If $i\in \{3, ..., k\}$, then applying the induction process, we get $\kappa_{G'_{i}}(u,v)=k-1$. Then, for all $1\leq i \leq k$ we have $\kappa_{G'_{i}}(u,v)=k-1$. Finally, since $G'_{i} \subset G'$, then $G'_{i}-S'\subset G'-S'$. Thus, $S'-\{x_{i}\} $ is a $uv$-separator of $G'_{i}$ as $S'$ is a $uv$-separator of $G'$, $x_{i}\in S'$ and $x_{i}\notin  G'_{i}$. Hence, $\kappa_{G'_{i}}(u,v)\leq |S'-\{x_{i}\}|=k-2$; which gives a contradiction.  \\
	\linebreak
	Let $C_{u}$ and $C_{v}$ be two connected components in $G'-S'$ such that $u\in C_{u}$ and $v\in C_{v}$. Since $S'$ is a $uv$-separator of $G'$ then, $C_{u}\cap C_{v}=\phi$. Set $S_{u}=S\cap C_{u}$ and $S_{v}=S\cap C_{v}$. Since $|S'|=k-1$, then $G-S'$ contains a $uv$-path. Then, there exists a connected component $C_{uv}$ in $G-S'$ containing both $u$ and $v$. We have $C_{uv}-e\subset G-S'-e=G'-S'$, and so $e$ separates $u$ and $v$ in $C_{uv}$, this implies that $e$ is a bridge of $C_{uv}$. Then $C_{uv}-e=C_{1}\cup C_{2}$ where $C_{1}$  and $C_{2}$  are two connected components containing $u$ and $v$ respectively. Without loss of generality, we may assume $x_{1} \in C_{1}$ and $x_{2}\in C_{2}$. We remark that $C_{1}=C_{u}$ and $C_{2}=C_{v}$. Then, $x_{1}\in S\cap C_{u}$ and $x_{2}\in S\cap C_{v}$, so $S_{u} \neq \phi$ and $S_{v}\neq \phi$.\\
	\linebreak
	Set $S_{1}=S_{u}$ and $S_{2}=S-S_{1}$, and let $P$ be a $uv$-path in $G-S_{2}$. Then, $P\cap S_{1} \neq \phi$. Define $x_{1}(P)\in V(P)$ such that $P_{[x_{1}(P),v]}\cap S_{1}=\{x_{1}(P)\}$. Clearly, $ x_{1}(P)\notin S_{2}$, and so $x_{1}(P)\notin C_{v}$; otherwise if $x_{1}(P)\in C_{v}$ and $x_{1}(P)\in S_{1}\subset S$, then $x_{1}(P)\in S_{v}\subset S_{2}$, which gives a contradiction. We have $ P_{[x_{1}(P),v]}\cap S'\neq \phi$; otherwise $P_{[x_{1}(P),v]}$ is a $x_{1}(P)v$-path in $G'-S'$, as $P\subseteq G-S_{2}\subseteq G'$ and $P_{[x_{1}(P),v]}\cap S'=\phi$, then $x_{1}(P)\in C_{v}$ which gives a contradiction. Define $x'_{1}(P)\in V(P)$ such that $x'_{1}(P)\in P_{[x_{1}(P),v]}\cap S'$. Clearly, $P_{[x'_{1}(P),v]}\cap S=\phi$, as $P\subseteq G-S_{2}$, and $P_{[x'_{1}(P),v]}\cap S_{1}=\phi$. \\
	\linebreak
	Define $S'_{1}=\{x'_{1}(P);P$ is a $uv$-path in $G-S_{2}\}$. Clearly $S'_{1}\cup S_{2}$ is a $uv$-separator of $G$, because all the $uv$-paths that do not have vertices in $S_{2}$, must have vertices in $S'_{1}$. Thus $|S|=\kappa_{G}(u,v)\leq |S'_{1}\cup S_{2}|=|S'_{1}|+|S_{2}|$, so $|S'_{1}|\geq |S|-|S_{2}|=|S_{1}|.$ \\
	\linebreak
	Similarly, for $P$ is a $uv$-path in $G-S_{1}$, we define $x_{2}(P),x'_{2}(P)\in V(P)$ such that $P_{[u,x_{2}(P)]}\cap S_{2}=\{x_{2}(P)\}$, $x'_{2}(P)\in P_{[u,x_{2}(P)]}\cap S'$. Define $S'_{2}=\{x'_{2}(P);P$ is a $uv$-path in $G-S_{1}\}$. The following properties are realized: $P_{[u,x'_{2}(P)]}\cap  S=\phi$ and $|S'_{2}|\geq |S_{2}|$. \\
	\linebreak 
	Since $|S'_{1}| \geq |S_{1}|$ and $|S'_{2}| \geq |S_{2}|$, then, $|S'_{1}|+|S'_{2}| \geq |S_{1}|+|S_{2}|=|S|=k$. Then, $S'_{1}\cap S'_{2}\neq \phi$; otherwise, $|S'_{1}|+|S'_{2}|=|S'_{1}\cup S'_{2}|\leq |S'|= k-1$ which gives a contradiction. Let $a \in S'_{1}\cap S'_{2} $. Then there exist $R$ and $Q$, two $uv$-paths in $G-S_{1}$ and $G-S_{2}$ respectively,  such that $a=x'_{2}(R)=x'_{1}(Q)$. So, $R_{[u,a]}\cup Q_{[a,v]}$ is a connected subgraph in $G$ that contains both $u$ and $v$, then this subgraph contains a $uv$-path $P$, and $P\cap S\subseteq (R_{[u,a]}\cup Q_{[a,v]})\cap S=\phi$; which gives a contradiction. Therefore, the desired result holds.
	
\end{proof} 
\section{A new proof of Menger's theorem}
\begin{lemma}
	\rm{
		Consider a graph $G$, and $u,v,x,y\in V(G)$ such that $uv\notin E(G)$ and $xy\in E(G)$. Suppose that $\kappa_{G-a}(u,v)=\kappa_{G}(u,v)-1$ for all $a\in E(G) \cup (V(G)-\{u,v\}) $. Let $N(y)=\{x_{0},x_{1},...,x_{t}\}$ with $x_{0}=x$. Set $$G'=G-y+\sum_{i=1}^{t}x_{0}x_{i}$$
		Then, $\kappa_{G'}(u,v)=\kappa_{G}(u,v)$.
	}
\end{lemma}
\begin{proof}
	In the beginning we must clarify that $t\geq 1$. Otherwise, we have $N(y)=\{x_{0}\}$, and so $y\notin V(P)$ for all $P$ being a $uv$-path in $G$. Thus, $\kappa_{G-y}(u,v)=\kappa_{G}(u,v)$, which gives a contradiction.\\
	\linebreak
	Set $\kappa_{G}(u,v)=k$. Let $H=G-xy$, and so $\kappa_{H}(u,v)=k-1$. Then there exists a $uv$-separator $S$ of $H$ such that $|S|=k-1$. Note that $x$ and $y\notin S$; otherwise, suppose without loss of generality that $x\in S$. Then, $G-S=G-x-S\subseteq H-S$. But $S$ is a $uv$-separator of $H$, and so $S$ is a $uv$-separator of $G$ satisfying $|S|=k-1$; which gives a contradiction. Set $S_{x}=S\cup \{x\}$. It is clear that $S_{x}$ is a $uv$-separator of $G$. \\
	\linebreak
	By the construction of $G'$ we have $G-y\subseteq G'$, then $\kappa_{G'}(u,v)\geq k-1$, as $\kappa_{G-y}(u,v)=k-1$. On the order hand, $G'-S_{x}=G-y+\sum_{i=1}^{t}xx_{i}-S_{x}\subseteq G-S_{x}$. Then, $S_{x}$ is a $uv$-separator of $G'$, similarly with $|S_{x}|=k$. Therefore, $\kappa_{G'}(u,v)\leq k$. So, $k-1\leq \kappa_{G'}(u,v)\leq k$. Suppose to the contrary that $\kappa_{G'}(u,v)\neq k$, then $\kappa_{G'}(u,v)=k-1$.\\
	\linebreak
	 Let $S'$ be a minimal $uv$-separator of $G'$, then $|S'|=k-1$. We have $x\notin S'$; otherwise,   $G'-S'=G-y +\sum_{i=1}^{t}xx_{i} - S'= G-(S'\cup \{y\})$, and $|S'\cup \{y\}|=k$ since $y\notin G'$. So, $S'\cup \{y\}$ is a minimal $uv$-separator of $G$ and $xy\in <S'\cup \{y\}>$. Thus, by using Theorem \ref{thm 1}, we have $\kappa_{G-xy}(u,v)=k$; which gives a contradiction.\\
	\linebreak
	Let $P$ be a $uv$-path in $G$. If $y\notin P$, then $P\subset G-y\subset G'$ and so $P\cap S'\neq \phi$ since $S'$ is a $uv$-separator of $G'$. If $y\in P$. Let $x_{i}$ and $x_{j}$ be the predecessor and successor of $y$ on $P$ respectively, $0\leq i\neq j \leq t$. If $x\in P$, then without loss of generality  suppose that $P_{[u,x]}\subset P_{[u,y]}$. Consider $P'=P_{[u,x]} \cup xx_{j} \cup P_{[x_{j},v]}$. Since $P'$ is a $uv$-path in $G'$, then $V(P')\cap S'\neq \phi $. So, $V(P)\cap S' \neq \phi$ since $V(P')\subset V(P)$. If $x\notin P$. Consider $P"=P_{[u,x_{i}]}\cup\{x\} \cup x_{i}x \cup xx_{j} \cup P_{[x_{j},v]}$. Similarly, $P"$ is a $uv$-path in $G'$, then $V(P")\cap S'\neq \phi$. Thus, $V(P)\cap S' \neq \phi$ since $V(P")-\{x\}=V(P)$ and $x\notin S'$. Therefore, $V(P)\cap S'\neq \phi$ in all the cases of $P$. This implies that $S'$ is a $uv$-separator of $G$ with $|S'|=k-1$; which gives a contradiction.  Therefore, the desired result holds.
	
\end{proof}

\begin{theorem} 
	\rm{\textbf{(Menger, 1927)}\\
	Consider a graph $G$, and $u,v\in V(G)$ such that $uv\notin E(G)$. Then the size of a minimal $uv$-separator of $G$ is equal to the maximum number of internally disjoint $uv$-paths in $G$; i.e. $\kappa_{G}(u,v)=\mu_{G}(u,v)$ .\\
	}
\end{theorem}

\begin{proof}
	Suppose that the statement is false, and let $G$ be a graph with the least number of vertices such that $\kappa_{G}(u,v)=k$ and $G$ contains no $k$ internally disjoint $uv$-paths. $G$ contains a spanning subgraph $H$ which has $\kappa_{H}(u,v)=k$ but $\kappa_{H-e}(u,v)=k-1$ for all $e\in E(H)$; $G=H$ possibly. Clearly, $\kappa_{G-a}(u,v) = k-1$ for all $a\in E(G)\cup (V(G)-\{u,v\})$.\\ 
	\linebreak
	\texttt{Claim.} There exists $x$ and $y\in V(G)-\{u,v\}$ such that $xy\in E(G)$.\\
	\par Suppose that for all $x,y\in V(G)-\{u,v\}$, $xy\notin E(G)$. Since $N(u)$ is a $uv$-separator, then $|N(u)|\geq k$. Set $\{w_{1},w_{2},...,w_{k}\}\subset N(u)$. If there exists $1\leq i \leq k$ such that $w_{i}\notin N(v)$ and since $V(G)-\{u,v\}$ is stable, then $N(w_{i})=\{u\}$, and so $w_{i}\notin P$ for all $P$ being a $uv$-path in $G$. Thus, $\kappa_{G-w_{i}}(u,v)=k$; which gives a contradiction to the fact that $\kappa_{G-w_{i}}(u,v)=k-1$. Then $w_{i}\in N(v)$ for all $1\leq i \leq k$. Then $\{uw_{i}v\}_{1\leq i \leq k}$ is a set of $k$ disjoints $uv$-paths in $G$; which gives a contradiction.\\
	\linebreak
	Let $x$ and $y\in V(G)-\{u,v\}$ such that $xy\in E(G)$. Let $N(y)=\{x_{0},x_{1},...,x_{t}\}$ with $x_{0}=x$. Set  $$G'=G-y+\sum_{i=1}^{t}x_{0}x_{i}.$$
	Using Theorem \ref{thm 1}, then $\kappa_{G'}(u,v)=k$. And since $v(G')<v(G)$, then $G'$ contains $k$ disjoint $uv$-paths, as $v(G)$ is minimal for a graph $G$ such that $\kappa_{G}(u,v)=k$ and $G$ contains no $k$ internally disjoint $uv$-paths. Let $\{P^1,P^2,..$ $.,P^k\}$ be the set of $k$ disjoint $uv$-paths in $G'$.\\
	\textbf{Case 1:} For all $1\leq i \leq t$, $xx_{i}\notin P^j$ $\forall 1\leq j\leq k$. Then $\{P^1,P^2,...,P^k\}$ is a set of $k$ internally disjoint $uv$-paths in $G$; which gives a contradiction.\\
	\textbf{Case 2:} There exists $j$, $1\leq j \leq k$ such that $xx_{i}\in E(P^j)$ for some $1\leq i \leq t$. Note that $d_{P^j}(x)=2$ since $x\notin \{u,v\}$. Without loss of generality, suppose that $x_{i}$ is the successor of $x$ on $P^j$, and let $w$ be the predecessor of $x$ on $P^j$.\\
	\begin{enumerate}
		\item If $w\neq x_{r}$, $\forall 1\leq r\neq i \leq t$. Consider $Q^j=P^{j}_{[u,x]}\cup \{y\} \cup xy \cup yx_{i} \cup P^{j}_{[x_{i},v]}$.\\
		\item If $w=x_{r}$, for some $1\leq r\neq i \leq t$. Consider $Q^j= P^{j}_{[u,x_{r}]}\cup \{y\}\cup x_{r}y \cup yx_{i}\cup P^{j}_{[x_{i},v]}$.
	\end{enumerate}
	In both cases $Q^{j}$ is a $uv$-path in $G$, and for all $1\leq s \neq j \leq k$, $x\notin V(P^{s})$ as $\{P^1,P^2,...,P^k\}$ is a set of internally disjoint $uv$-paths, so $P^s\subseteq G-y\subseteq G$. Clearly $\{P^1,..,P^{j-1},Q^j,P^{j+1},..,P^k\}$ is a set of $k$ disjoint $uv$-paths in $G$, a contradiction  and so the result holds.
	
\end{proof}


\begin{thebibliography}{99}

    \bibitem{Dirac}  G. A. Dirac, Short proof of Menger's graph theorem. Mathernatika 13 (1966) 42-44. 
	
	\bibitem{Diestel} R. Diestel, Graph Theory, Springer, New York, (1997) pp. 50-55.
	
	
	\bibitem{Fulkerson} L. R. Ford, Jr. and D. R. Fulkerson, Maximal flow through a network, Canad. J. Math. 8 (1956), 399-404. 
	
	\bibitem{Goring} F. G\"oring, Short Proof of Menger's Theorem, Discrete Mathematics 219 (2000) 295-296.
	
	\bibitem{McCuaig} W. McCuaig, A simple proof of Menger's theorem, J. Graph Theory 8 (1984) 427-429. 
	
	\bibitem{Menger}  K. Menger, Zur allgemeinen Kurventheorie, Fund. Math. 10 (1927) 96-115.
	
	\bibitem{Peter} P. V. O'Neil, A new proof of Menger's theorem, J. Graph Theory 2 (1978) 257-259.
	
	
	
	

	
	
	
	
	
	
	
\end{thebibliography}
\end{document}